\documentclass[a4paper,12pt]{extarticle}
\usepackage[utf8]{inputenc}

\usepackage{geometry}

\usepackage{amsmath,graphicx, amssymb,color,ulem,bbm,geometry}

\usepackage{amsthm}

\usepackage{float}
\usepackage{epsfig}
\usepackage{enumitem}

\newtheorem{theorem}{Theorem}
\newtheorem{Proposition}[theorem]{Proposition}

\usepackage{tikz}
\usetikzlibrary{shapes.geometric, arrows}
\usetikzlibrary{arrows.meta}

\tikzstyle{sus} = [rectangle, text width=1cm, minimum height=1cm,text centered, draw=black, fill=blue!20]
\tikzstyle{asym} = [rectangle, text width=1cm, minimum height=1cm,text centered, draw=black, fill=orange!40]
\tikzstyle{sym} = [rectangle, text width=1cm, minimum height=1cm,text centered, draw=black, fill=pink!70]
\tikzstyle{rec} = [rectangle, text width=1cm, minimum height=1cm,text centered, draw=black, fill=green!20]
\tikzstyle{exp} = [rectangle, text width=1cm, minimum height=1cm,text centered, draw=black, fill=gray!30]
\tikzstyle{vacc} = [rectangle, text width=1cm, minimum height=1cm,text centered, draw=black, fill=white!50]

\tikzstyle{arrow} = [thick,->,>=stealth]

\title{Assessing the Impact of (Self)-Quarantine Through a Basic Model of Infectious Disease Dynamics}

\author{J\'{o}zsef Z. Farkas, Roxane Chatzopoulos \\ 
Computing Science and Mathematics \\ University of Stirling, Stirling FK9 4LA, UK \\ jozsef.farkas@stir.ac.uk}

\date{\today}

\begin{document}

\newgeometry{left=1.95cm, right=1.95cm, top=2.5cm, bottom=3cm}

\maketitle

\abstract{We introduce a system of differential equations to assess the impact of (self-)quarantine of symptomatic infectious individuals on disease dynamics. To this end we depart from using the classic bilinear infection process, but remain still within the framework of the mass-action assumption. From the mathematical point of view our model is interesting due to the lack of continuous differentiability at disease free steady states, which implies also that the basic reproductive number cannot be computed following established approaches for certain parameter values. However, we parametrise our mathematical model using published values from the COVID-19 literature, and analyse the model simulations. We also contrast model simulations against publicly available COVID-19 test data focusing on the first wave of the pandemic during March - July 2020 in the UK. Our simulations indicate that actual peak case numbers might have been as much as 200 times higher than the reported positive test cases during the first wave in the UK. We find that very strong adherence to self-quarantine rules yields (only) a reduction of 22$\%$ of peak numbers and delays the onset of the peak by approximately 30-35 days. However, during the early phase of the outbreak the impact of (self)-quarantine is much more significant. We also take into account the effect of a national lockdown in a simplistic way by reducing the effective susceptible population size. We find that in case of a 90$\%$ reduction of the effective susceptible population size, strong adherence to self-quarantine still only yields a 25$\%$ reduction of peak infectious numbers when compared to low adherence. This is due to the significant number of asymptomatic infectious individuals in the population.}

\section{Introduction}

The recent COVID-19 outbreak has reinvigorated interest in the use of mathematical models by the infectious disease community, i.e., to make predictions and assess various hypotheses surrounding infectious disease dynamics (in particular COVID-19). One important question that has been investigated is the role and proportion of asymptomatic infected individuals in the population, with regard to disease dynamics. This question is particularly important when studied in relation to assessing the impact of (self)-quarantine, as clearly without readily available mass testing (as was the case for COVID-19 in early 2020 around most parts of the world), asymptomatic individuals may continue to mix in the population, and only symptomatic infectious individuals will/may self-quarantine. 

Despite the reinvigorated interest in mathematical models of infectious disease dynamics, very few authors ventured to use some novel models, in particular when it comes to the modelling of the infection process itself; that is how new infections arise from the contact between susceptible and infected individuals. Apparently almost all recent COVID-19 modelling work uses the classic bilinear infection process, see, e.g., \cite{Danon,Ghostine2021,Webb2021-JTB,Webb2020-MBE,Webb2020-IDM,Mwalili2020}; which dates back to Kermack and McKendrick \cite{Kermack1927,Kermack1932}. The idea behind the classic $\lambda\,SI$ infection term in the early works of Kermack and McKendrick is the so-called mass-action assumption, s\mbox{ee e.g., \cite{Diekmann-Heesterbeek,Heesterbeek}.} It is hypothesized that every susceptible individual has an equal chance to meet with any of the infected individuals, hence the number of possible contacts (\mbox{per unit time}) of a single susceptible individual is proportional to $I$, which is a linearly increasing function of the total infected population size $I$. However, it has been recognised for a while that the contact network of the population is not homogeneous; still, it can be shown that in large enough populations with sufficient mixing, ODE systems may provide good enough average information of the real (stochastic and individual based) process, see, e.g., \cite{Simon,Sharkey}. Another major weakness of the classic $SIR$ model is that its (long-term) behaviour only depends on the initial number of susceptible individuals, and it is independent of the infectious ones. 
It is worthwhile to note that the mass-action assumption originates from studying chemical reaction networks and indeed it has been criticised, see, e.g., \cite{Heesterbeek} (also for a complete historical overview). In particular it is acknowledged (see, e.g., \cite{Heesterbeek}) that the forceful paradigm of using a simple bilinear infection term of the form $\lambda\, SI$ is not necessarily justified from the modelling point of view and its widespread use is rather due to its simplicity.

On the other hand, one easily finds models in the mathematical literature, which contain different nonlinear infection terms. One of the earliest of such papers is by \cite{Capasso1978}, but see also the much more recent ones \cite{Farkas, Korobeinikov2006}. 
There are clear arguments that, in many situations, it will be much more realistic to use an infection term of the form $S\,g(I)$, where, for example, $g$ saturates for large values of $I$. For example it is clear that during a pandemic as the number of infected individuals increase in the population susceptibles will naturally reduce the number of their contacts inducing a saturation effect in $g$. In the next section, we will introduce a mathematical model incorporating a new nonlinear infection process; that is the number of contacts will not be a linearly increasing function of the infectious population size. Partially motivated by the recent COVID-19 pandemic, we focus on two key \mbox{aspects: }

\begin{enumerate}
\item The impact of asymptomatic infectious individuals on disease dynamics.
\item The impact of (self)-quarantine on disease dynamics.
\end{enumerate} 

While the first issue can be simply studied by separating the group of infectious individuals into two compartments; the second issue, we argue, can be handled via the introduction of a novel infection term. The key finding we obtain using our novel mathematical model is that the impact of (self)-quarantine is significant only during the very early phase of an outbreak, and its impact significantly diminishes over time. We argue that this effect is due to the presence of a significant number of asymptomatic infectious individuals. This finding is rather interesting in light of the transmission term we use, as the saturation effect in disease transmission increases with the population density of infectives in our model. From the practical point of view, our finding is important in light of the extensive government measures introduced in response to COVID-19 in the UK. We also contrast our model simulation outputs against publicly available COVID-19 test data, focusing on the first wave of the pandemic in March--July 2020 in \mbox{the United Kingdom.} 

The number of recently published papers using systems of differential equations to model COVID-19 disease dynamics is enormous. Attempts have been made to design models, which separate individuals who are known to be infected from those that are untested/unreported, see, e.g., \cite{Webb2020-MBE,Webb2021-JTB} and also \cite{Webb2020-IDM}. It is worthwhile to note though that interestingly enough in \cite{Webb2020-MBE,Webb2021-JTB,Webb2020-IDM} it is assumed that the transmission rate for symptomatic and asymptomatic individuals are the same. Early on in the pandemic, it was also hypothesized that environmental transmission was significant, see, e.g., \cite{Mwalili2020}; however, later studies confirmed that environmental variations cannot explain varying transmission rates see \mbox{e.g., \cite{Poirier2020}}. What most researchers would agree on is that quarantine of symptomatic individuals may play a significant role in infectious disease dynamics, in particular in conjunction with significantly different transmission rates for asymptomatic versus symptomatic individuals (see, e.g., \cite{Ferguson}); and these are the two issues we will focus on here.

\section{Materials and methods}

Our goal is to introduce and study a basic compartmental model, which can be considered as an extension/modification of a classic $SEIR$ model in two major directions: 
\begin{enumerate}
\item We introduce a compartment for the asymptomatic infectious individuals.
\item We model the effect of (self)-quarantine via the introduction of a new nonlinear (sublinear to be precise) transmission term.
\end{enumerate} 
Thus, in our mathematical model, at any given time, individuals will belong to one of the following 5 compartments.
\begin{itemize}
    \item $S$---susceptible;
    \item $E$---exposed;
    \item $A$---asymptomatic infectious; 
    \item $I$---symptomatic infectious;
    \item $R$---recovered/removed.
\end{itemize} 

To arrive at a tractable mathematical model (and simultaneously to remain relatively close in nature to some of the existing compartmental models, for comparison purposes) we impose the following assumptions:
\begin{enumerate}
    \item Newly infected individuals all enter the exposed compartment first due to a latency period. It is clear from the studies that, for COVID-19, there is a significant incubation period of approximately $7.73$ days; see \cite{Qin}. 
    \item Exposed individuals spend on average in the $E$ compartment $\alpha^{-1}$ time units (where, e.g., for COVID-19 $\alpha^{-1}=7.73$ days, according to \cite{Qin}), after which a  proportion: $0<p<1$ of them becomes symptomatic infectious, while the remaining individuals become asymptomatic infectious. 
    \item Both asymptomatic and symptomatic individuals may pass on the disease, but naturally (as is the case for all influenza-like diseases, spread mainly by droplets, e.g., COVID-19) the transmission rate is significantly higher for symptomatic individuals: $\beta_I$, when compared to the transmission rate for asymptomatic infected individuals: $\beta_A$; that is $0<\beta_A\ll \beta_I$; still asymptomatic infected individuals may also pass on the disease (for example via very close contact).
    \item We take into account the effect of (self)-quarantine; that is, we assume that at any given time a subset of the  symptomatic infected individuals (self)-quarantine (this is supported by reports, e.g., during the COVID-19 pandemic, see, e.g., \cite{Bodas,Smith}). We use a power law to model this, instead of using, for example, a fixed proportion, as it is known that (online) human interaction activity will impact adherence rates, and such activity can be often approximated by power laws (see, e.g., \cite{Muchnik, Rybski}). Moreover, it may be natural to assume that the impact of quarantine will be more significant for relatively higher infectious population densities. Thus, in combination with the classic mass action assumption, we propose to use an infection term of the following form:
    $$\beta_A\, S(t)A(t)+\beta_I\, S(t)(I(t))^\kappa,\quad 0\le \kappa\le 1.$$
Note that, for smaller values of $\kappa$, the transmission (for example contact) between susceptibles and symptomatic infectious individuals is reduced. In particular, $\kappa=0$ would correspond to a complete quarantine of the symptomatic infectious population.    
    \item We do not incorporate population dynamics since we want to focus on the disease dynamics over a short period of time here (that is, for COVID-19, a 4-month period during March--July 2020).  Moreover, we do not explicitly incorporate disease-induced mortality into our model, although it can be understood that deceased people have entered the removed compartment.
\end{enumerate}

The assumptions above allow us to introduce the following basic model

\begin{equation}\label{model}
\begin{aligned}
    S'(t)= & -\beta_A\, S(t)A(t)-\beta_I\, S(t)(I(t))^\kappa \\
    E'(t)= & \beta_A\, S(t)A(t)+\beta_I\, S(t)(I(t))^\kappa-\alpha\,E(t) \\
    A'(t)= & (1-p)\,\alpha\,E(t)-\gamma_A\, A(t) \\
    I'(t)= &  p\,\alpha\,E(t)-\gamma_I\, I(t)\\
    R'(t)= & \gamma_A\,A(t) +\gamma_I\,I(t)
\end{aligned}
\end{equation}
naturally quipped with an initial condition of the form
\begin{equation*}
    \left(S(0),E(0),A(0),I(0),R(0)\right):=\left(S_0,E_0,A_0,I_0,R_0\right)\in\mathbb{R}^5_+,\quad S_0+E_0+A_0+I_0+R_0=N.
    \end{equation*}

The compartmental model \eqref{model} above can be considered an extension of the classic Kermack--McKendrick $SIR$ model (\cite{Kermack1927,Kermack1932}), in two major directions: we introduce an asymptomatic infectious compartment, and we model the effect of (self)-quarantine via the introduction of a nonlinear infection process. 
We note that the model above can be considered a starting point of departure from the classic Kermack--McKendrick SEIR 
model, and there are various natural modifications (such as the introduction of time-dependent parameters), which would make it widely  applicable; we discuss some of these extensions in the last section.
For the readers' convenience, we present a simple diagram below describing our model.

\begin{equation*}
\begin{tikzpicture}[node distance=2cm]
<TikZ code>
\node (susceptibles) [sus] {S};
\node (exposed) [exp, right of=susceptibles, xshift=2cm] {E};
\node (symptomatic) [sym, above right of=exposed, xshift=2cm] {I};
\node (asymptomatic) [asym, below right of=exposed, xshift=2cm] {A};
\node (recovered) [rec, above right of=asymptomatic, xshift=2cm] {R};

\draw [arrow] (susceptibles) -- node[anchor=south] {$(\beta_A,\beta_I)$} (exposed);
\draw [arrow] (exposed) -- node[anchor=south] {$p\,\alpha$} (symptomatic);
\draw [arrow] (exposed) -- node[anchor=north] {$(1-p)\,\alpha $} (asymptomatic);
\draw [arrow] (symptomatic) -- node[pos=0.5,above] {$\gamma_I$} (recovered);
\draw [arrow] (asymptomatic) -- node[pos=0.5,below] {$\gamma_A $} (recovered);

\end{tikzpicture}
\end{equation*}
    
Note that (because of assumption 5 above) we have
\begin{equation*}
\frac{d}{dt}\left(S(t)+E(t)+A(t)+I(t)+R(t)\right)=0,\quad \forall t>0,
\end{equation*}
that is the total population size (denoted by $N$ below) is preserved
\begin{equation*}
S(t)+E(t)+A(t)+I(t)+R(t)=N>0, \quad \forall t\ge 0.
\end{equation*}

Due to the lack of Lipschitz continuity at $I=0$, we establish existence and uniqueness (and positivity) of solutions (when $I\neq 0$) below.

\begin{Proposition}
For every initial condition $(S(0),E(0),A(0),I(0)\neq 0,R(0))\in \mathbb{R}_+^5$ model \eqref{model} admits a unique solution, which remains non-negative for all times.
 \end{Proposition}
\begin{proof}[Proof of Proposition 1]
Let us introduce the notation ${\bf u}(t)=(S(t),E(t),A(t),I(t),R(t))^t$ and rewrite model \eqref{model} in matrix form as 
\begin{equation}
{\bf u}'(t)= H({\bf u}(t)),\quad {\bf u}(0)=:{\bf u}_0\in\mathbb{R}^5,
\end{equation}
where we define 
\begin{equation}
H({\bf u}):=\begin{pmatrix}
H_1({\bf u})= & -\beta_A\,SA-\beta_I\, SI^\kappa \\
H_2({\bf u})= & \beta_A\, SA+\beta_I SI^\kappa-\alpha\, E \\
H_3({\bf u})= &(1-p)\alpha\, E-\gamma_A\, A \\
H_4({\bf u})= & p\alpha\, E-\gamma_I\,I \\
H_5({\bf u})= &\gamma_I\, I+\gamma_A\, A
\end{pmatrix}.
\end{equation}

Next, we note that $H$ is locally Lipschitz continuous on the open set 
\begin{equation*}
\mathbb{O}:=\mathbb{R}^5\setminus \{(S,E,A,0,R)\in\mathbb{R}^5\},
\end{equation*} 
hence for any ${\bf u}(0)\in\mathbb{O}$ a unique local solution of \eqref{model} exist by the Picard-Lindel\"{o}f theorem. 
To show that solutions exist globally, it is sufficient to show (see, e.g., Corollary 2.5.3 in \cite{Pruss2010}) that there exists a constant $\omega\ge 0$, such that 
\begin{equation}\label{global-criterion}
\left\langle H({\bf w}),{\bf w}\right\rangle\le \omega\, ||{\bf w}||_2^2,
\end{equation}
for all ${\bf w}\in\mathbb{O}$. Above in \eqref{global-criterion} $\langle \cdot,\cdot\rangle$ stands for the usual inner product on $\mathbb{R}^5$, and $||\cdot||_2$ for the standard Euclidean norm on $\mathbb{R}^5$. We have 
\begin{equation}
\begin{aligned}
\left\langle H({\bf w}),{\bf w}\right\rangle = & (-\beta_A\,SA-\beta_I\, SI^\kappa)S+(\beta_A\, SA+\beta_I SI^\kappa-\alpha\, E)E \\
& +((1-p)\alpha\, E-\gamma_A\, A)A+(p\alpha\, E-\gamma_I\,I)I+(\gamma_I\, I+\gamma_A\, A)R \\
\le & \max\{\beta_A,\beta_I,\alpha, \gamma_A,\gamma_I\}N^2\,(S^2+E^2+A^2+I^2+R^2) \\
= & \max\{\beta_A,\beta_I,\alpha, \gamma_A,\gamma_I\}N^2\,||{\bf w}||_2^2,
\end{aligned}
\end{equation}
using that we have $S,E,A,I,R\le N$. Finally we note that we have 
\begin{equation}
H_1({\bf w})|_{S=0}=0,\,\, H_2({\bf w})|_{E=0}\ge 0,\,\, H_3({\bf w})|_{A=0}\ge 0,\,\, H_4({\bf w})|_{I=0}\ge 0,\,\, H_5({\bf w})|_{R=0}\ge 0, 
\end{equation}
hence, solutions starting in $\mathbb{O}_+$ will remain non-negative for all times, see \cite{Nagumo1942}.
\end{proof}

With regards to further mathematical properties we note that model \eqref{model} admits a family of steady states of the form: 
\begin{equation*}
\left(S_*,0,0,0,R_*\right),\quad S_*+R_*=N,\quad S_*,R_*\ge 0,
\end{equation*} 
(a line in the phase space). On the other hand, from the mathematical point of view it is worthwhile to note that the right hand side of model \eqref{model} is not Lipschitz continuous when $I=0$, for $0<\kappa<1$; hence, for example, one cannot linearise the model at its non-trivial steady states for these $\kappa$ values. Note that though steady states clearly exist, as only uniqueness of solutions may be lost when $I=0$, but clearly we are not interested in such solutions. We make it clear that we only intend to use model \eqref{model} over a relatively short time-scale, e.g., during the first wave of the UK pandemic during March--July 2020 (after which, all compartmental models tend to deviate from the real process, as they are rather simplistic); hence, asymptotic properties of the model are irrelevant; although we refer the interested reader to section 4 for more details on qualitative properties of \mbox{model \eqref{model}} for $\kappa=1$, and in particular in comparison to the classic Kermack--McKendrick model.

\section{Results}

In this section, we parametrise our model and analyse model simulations, focusing on the first wave of COVID-19 (pre Delta variant) in the UK during March--July of 2020.  In particular, to parametrise our model \eqref{model}, we searched the COVID-19 literature to find some realistic values, although it became quickly clear that there are significant variances of estimates one finds in different studies. In our model \eqref{model}, time is measured in days, and as already previously mentioned, we take an incubation period of $7.73$ days, which yields $\alpha=0.129$. For the total UK population, we used the 2020 estimate of $N=67,081,000$ by the Office for National Statistics (see \cite{ONS}). It is worthwhile to note that it is very difficult to find any reliable estimates for transmission probabilities for COVID-19. On the other hand, there are a lot of published studies estimating the difference between transmission rates of asymptomatic versus symptomatic infectious individuals. For example, \cite{Ferguson} estimates that symptomatic individuals will have $50\%$ higher transmission rates. Therefore, in our simulations, initially, we set $\beta_A=2 N^{-1}$ and $\beta_I=3 N^{-1}$. Note that, for example, $\beta_A$ is the per (susceptible) capita contact rate per unit of time (day), times the transmission probability. So, for example, the value we choose for $\beta_A$ would correspond to $5$ contacts per unit time (day), on average, with a transmission probability of $0.4$ (upon contact). Naturally one would assume that the transmission probability is much higher for symptomatic infectious individuals. However, the average number of contacts might be reduced; hence, the $50\%$ higher transmission rate for symptomatic individuals (also used in \cite{Ferguson}). Several studies (see, e.g., \cite{Bullard,vanKampen}) reported no or low detectable virus levels 8 days from the onset of symptoms; hence, we take the recovery rate $\gamma_I=0.125$. Since there is no reliable evidence that the recovery rate would be different for asymptomatic cases, we take the same value for $\gamma_A$. There are a number of studies with widely varying figures ($1\%$ to $78\%$)  concerning the proportion of asymptomatic versus symptomatic cases. This is not surprising as asymptomatic individuals are often not tested; hence, it is very difficult to obtain reliable estimates for their numbers. A recent study \cite{Alene} reported that around the third of cases are asymptomatic; hence, we take $p=0.66$ as a baseline figure for our simulation setups. Note that, in our model, $p$ is the probability to become symptomatic, and not exactly the proportion of symptomatic individuals among the infected ones at any given time (which will also depend on the initial conditions used), we will discuss this further later on when interpreting the simulation results. The parameter $\kappa$ measures adherence levels to quarantine and, hence, the reduction of symptomatic transmission of the disease, which we are going to vary to assess its impact in early and late stages of the first wave of the pandemic in the UK.  Note that a low value of $\kappa$ means strong adherence to self-quarantine rules. It is also important to  note that $\kappa$ (naturally) is a dimensionless (scalar) quantity.

We contrast our model simulation results against publicly available data in \cite{World}, focusing on the first wave (pre Delta variant), March--July 2020 of the COVID-19 pandemic in the UK. The main issue with most of the publicly available COVID-19 data we found is that usually only the number of new positive tests and deaths were reported on a daily basis. However, our simulation output is the total number of currently infectious symptomatic and asymptomatic individuals, which does not allow for a direct comparison. It is clear that there is even more ambiguity when we take into account that, for anyone tested (and reported) positive, it is not known when exactly they acquired the infection; that is, there is a variation of the days elapsed from the onset of symptoms until a positive test is recorded. It was not our focus here to perform a rigorous data analysis. Instead, to estimate the number of current infectious individuals for a given day, we took a backward time window and added the total number of new cases reported, and subtracted the number of reported deaths over a two-sided time window spanning 19 days. The value 18 to 19 days from diagnosis to death was reported in a recent study from Australia \cite{Marschner}, and we do not have any reason to believe that this would be significantly different for the UK.

\begin{figure}[H]
\begin{centering}
\includegraphics[width=15.5 cm, height=12.5cm]{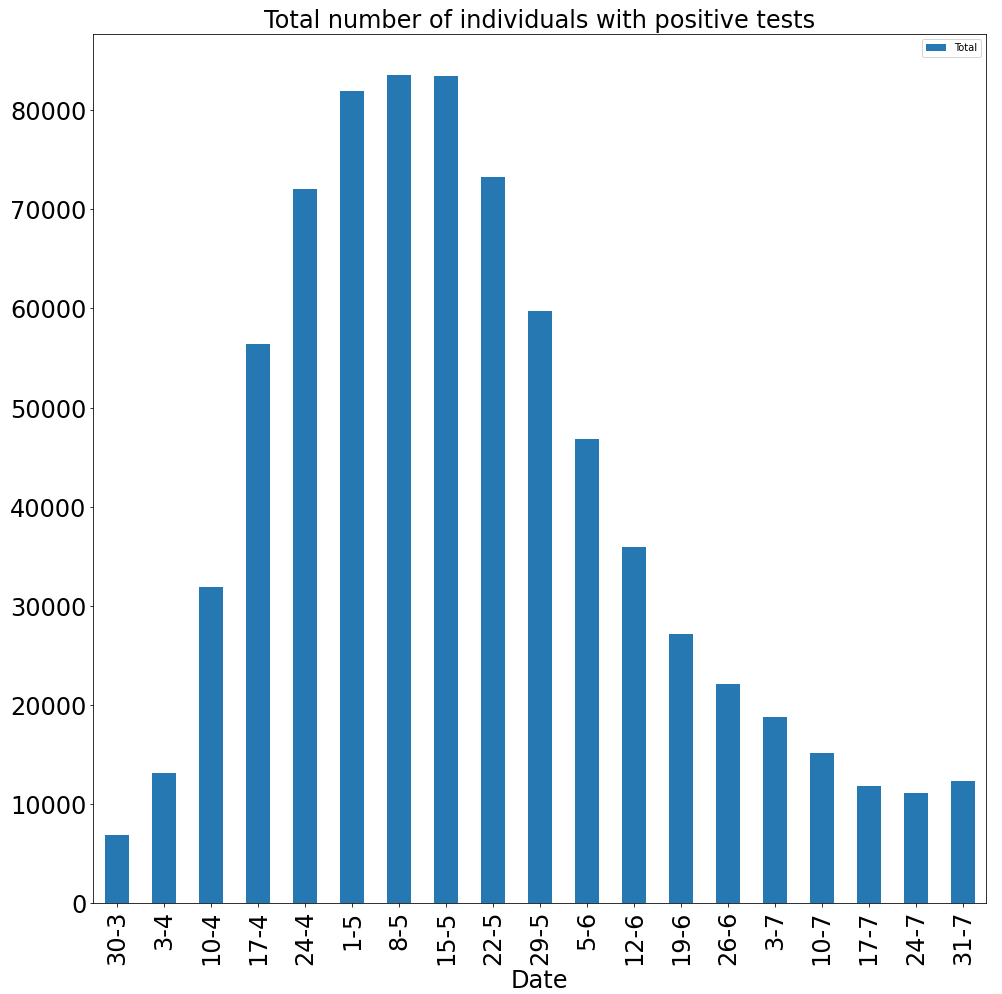}
\end{centering}
\caption{Estimate for the total number of individuals with positive tests at weekly time instances during March - July 2020 in the UK, obtained from the data sources in \cite{World}. \label{data}}
\end{figure}

It is clear that, due to insufficient testing during the first wave (March-July 2020) in the UK, the actual number of infectious  (asymptomatic and symptomatic) individuals was likely to be much higher than the test data indicates (this is also confirmed by the volume of positive tests recorded during the second wave, as much as 15 times higher, when testing was readily available). Our simulations in fact indicate that the actual number of cases could have been as much as 200 times of that reported via testing (see Figure \ref{fig-setup1-total} below). It is also clear that the significant national lockdown significantly reduced the total number of cases. We are going to investigate this by reducing the "available" (effective) susceptible population size significantly in Setup 1---reduced $S(0)$, see below.  The shape of the curves (note that 1 March corresponds to time 0 in our simulations) indicate very good agreement both for Setup 1 and Setup 1 with reduced susceptibles, see, e.g., Figure \ref{fig-setup1-total} below. However, the peak number of individuals in our simulation Setup 1 is around 85 to 90 days, which coincides with the peak number of the data based on positive tests. In contrast, in Setup 1, with reduced susceptibles, the peak appears to be at 75 days, which may suggest that despite the national lockdown during March--June 2020, the possible number of susceptibles coincided with the actual population size. This could be explained with the high density network structure coupled with the fact that significant layers of the population were still working and mixing, etc.

\begin{figure}[H]
\includegraphics[width=13.5 cm]{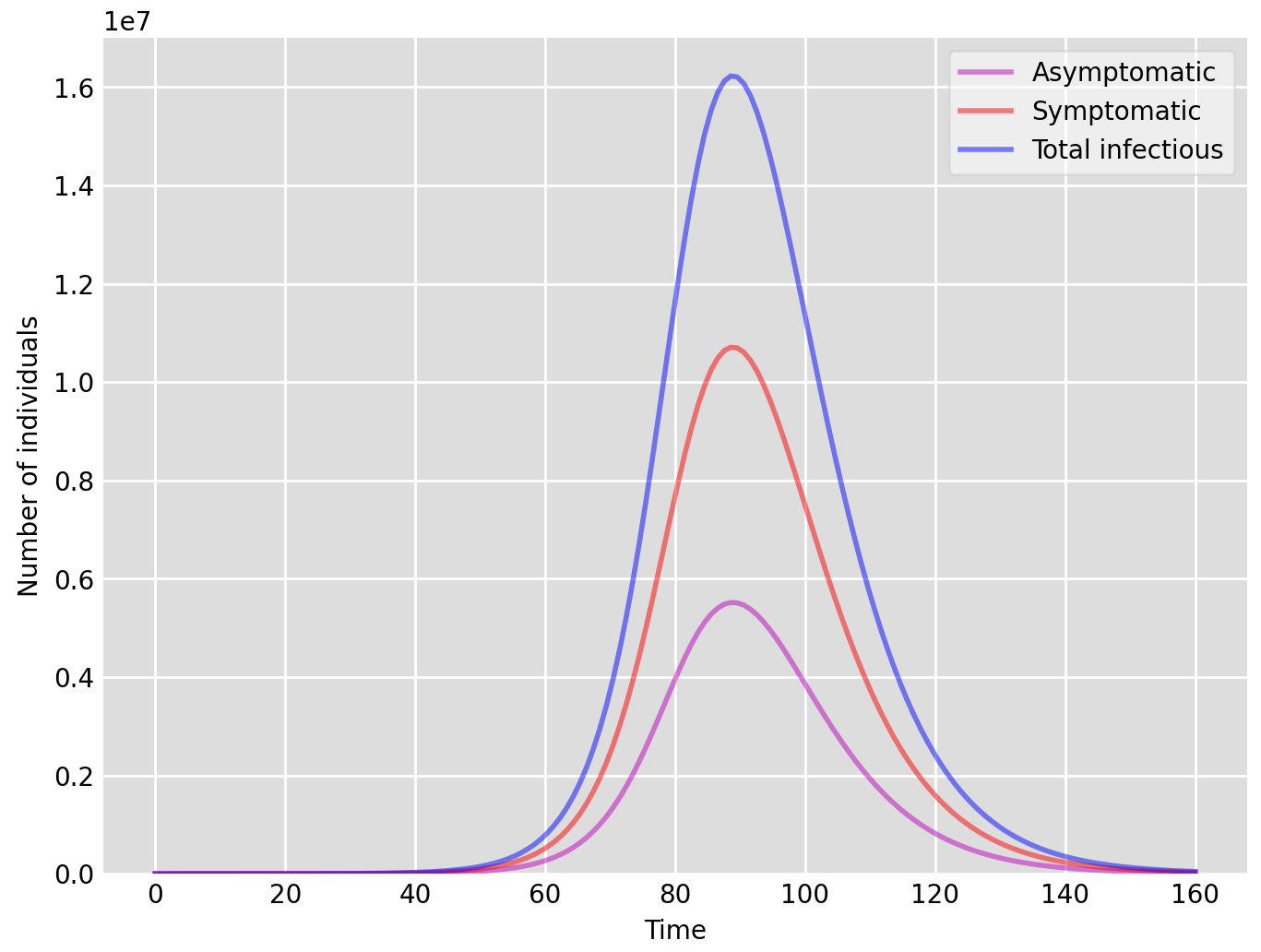}
\caption{Setup 1 - model simulation results for the total number of individuals $A(t),I(t)$  and $A(t)+I(t)$. \label{fig-setup1-total}}
\end{figure}   

Below, we present some key model simulation results. We used discretisation to produce solutions of model \eqref{model} in Python. In Setup 1 and 2, the initial condition we used consisted of the whole susceptible UK population and a very small number of exposed, asymptomatic infectious and symptomatic infectious individuals ($50,\, 10$ and $40$, respectively). In Setup 1 and 2---reduced $S(0)$, we reduced the initial susceptible population by $90\%$, and changed, accordingly, the parameter values $\beta_A$ and $\beta_I$, while the other parameters were kept the same; see Table \ref{tab1} below. This means that, we hypothesize that, in reality, there is an effective susceptible population (which in the simulation is $10\%$ of the total population), and in fact a significant proportion of susceptible individuals are isolated. From the practical point of view, the reduction of the susceptible population could be a result of a number of factors, for example a national lockdown, shielding, self-isolation, etc.

\begin{table}[H] 
\caption{We summarise the model parameters used for the simulations in Python.\label{tab1}}
\begin{tabular}{l|c|c|c|c|c|c|c}
\textbf{Model parameters}	& $\beta_A$	& $\beta_I$ & $\gamma_A$ & $\gamma_I$ & $\alpha$  & $p$ & $\kappa$ \\
\hline
\hline
Setup 1		&  $\frac{2}{67081000}$ & $\frac{3}{67081000}$ &  $0.125$  & $0.125$  & $0.129$  & $0.66$ & $0.1$ \\
\hline
Setup 2		&  $\frac{2}{67081000}$ & $\frac{3}{67081000}$ &  $0.125$  & $0.125$  & $0.129$  & $0.66$ & $0.9$ \\
\hline
Setup 1 - reduced $S(0)$		&  $\frac{2}{6708100}$ & $\frac{3}{6708100}$ &  $0.125$  & $0.125$  & $0.129$  & $0.66$ & $0.1$ \\
\hline
Setup 2 - reduced $S(0)$		&  $\frac{2}{6708100}$ & $\frac{3}{6708100}$ &  $0.125$  & $0.125$  & $0.129$  & $0.66$ & $0.9$ \\
\end{tabular}
\end{table}

We highlight some of our key findings from the model simulations, and in light of the test data, as follows:

\begin{enumerate}
\item Setup 1 ($\kappa=0.1$) simulation results indicate very good agreement between the dates of the peak number of positive tests from the data and model simulation output for the total number of cases. In particular, test data indicate a peak at around 63 to 77~days (see Figure \ref{data}), while simulation outputs show infectious case numbers peaking at around 85 days (see Figure \ref{fig-setup1-total}), and this is what one naturally expects as someone who tests positive may remain infectious for a further 5 to 10 days. While the  wave takes place over a much shorter timescale in Setup 2, infectious individual numbers peak at around days 50--55 (see Figure \ref{fig-setup2-total}), quarantining of a significant proportion of symptomatic individuals obviously slows down the spread of the disease. Hence, in principle, our simulation outputs would indicate that symptomatic infectious individuals adhered to self-quarantine rules during the first wave of the pandemic in the UK (indicated by the agreement with Setup 1).
\item Comparing the peak number of cases in Setup 1 of $1.64\times10^7$ (see Figure \ref{fig-setup1-total}) to that against Setup 2, which is $2\times10^7$ (see Figure \ref{fig-setup2-total}), we can see that a reduction of approximately $22\%$ of peak numbers was achieved by moving from a very loose \mbox{$(\kappa=0.9$)} to a very strong ($\kappa=0.1$) adherence to self-quarantine of symptomatic infectious individuals. Keeping in mind that we assumed that (a significant) two-thirds of exposed individuals become symptomatic, in our opinion, this is not a drastic change as one might hope for. Indeed, if for example only a third of exposed individuals become symptomatic, then our model \eqref{model} predicts that the impact of self-quarantine of symptomatic infectious individuals on peak case numbers would be negligible. However, self-quarantine has a much more significant impact during the early phase of the outbreak, which is due to the delay of the onset of the peak.
\item Based on the literature, we choose the parameter value $p=0.66$, meaning that, on average, two-thirds of the exposed individuals become symptomatically infectious. In both Setup 1 and 2, we observe that $\frac{A(t)}{I(t)}$ tends to the same constant $0.51\dot{5}\dot{1}$ as time goes to infinity, which is realistic, as this should be specific to the disease. Note that $0.51\dot{5}\dot{1}=\frac{0.34}{0.66}$ as one would expect, and simulations show that the ratio $\frac{A(t)}{I(t)}$ stabilizes very quickly. 
\item  When comparing Figures \ref{Setup1-proportions} and  \ref{Setup2-proportions}, we can see that the proportions $\frac{E(t)}{A(t)}$ are drastically different; this is due to the increased infection pressure in Setup 2, with low adherence to self-quarantine of symptomatic infectious individuals.
\item Comparing changes (Setup 1 to Setup 1 reduced $S(0)$ vs. Setup 2 to Setup 2 reduced $S(0)$), we can conclude that, in both cases, naturally ,there is approximately a $90\%$ reduction in the total number of infectious cases ($A(t)+I(t)$), due to the $90\%$ reduction in $S(0)$, the effective susceptible population, (compare Figure \ref{fig-setup1-total} to Figure \ref{fig-setup1reduced-total}). Infectious individual number peaks are also shifted (to an earlier date) by 13 days vs. 10 days in Setup 1 vs. Setup 2 (compare Figure \ref{fig-setup1-total} to Figure \ref{fig-setup1reduced-total} and Figure \ref{fig-setup2-total} to \mbox{Figure \ref{fig-setup2reduced-total}).} Similarly, the reduction in the peak number of cases from Setup 1 reduced $S(0)$ to that of Setup 2 reduced $S(0)$ is approximately $25\%$ (compare Figure \ref{fig-setup1reduced-total} to Figure \ref{fig-setup2reduced-total}), which is comparable to the reduction of $22\%$ from Setup 1 to Setup 2. Importantly, these simulation results may indicate that there is no significant combined effect of a national lockdown (modelled by the reduced effective susceptible population) and strong adherence to self-quarantine rules for symptomatic infectious individuals, when there is a significant proportion of asymptomatic infectious individuals ($\approx$34$\%$) and no mass testing, as was the case during the first wave of the pandemic in the UK. It is also clear from the simulations that all of the measures (e.g., national lockdown, (self)-quarantine) prolong the pandemic. 
\end{enumerate}

\begin{figure}[H]
\includegraphics[width=13.5 cm]{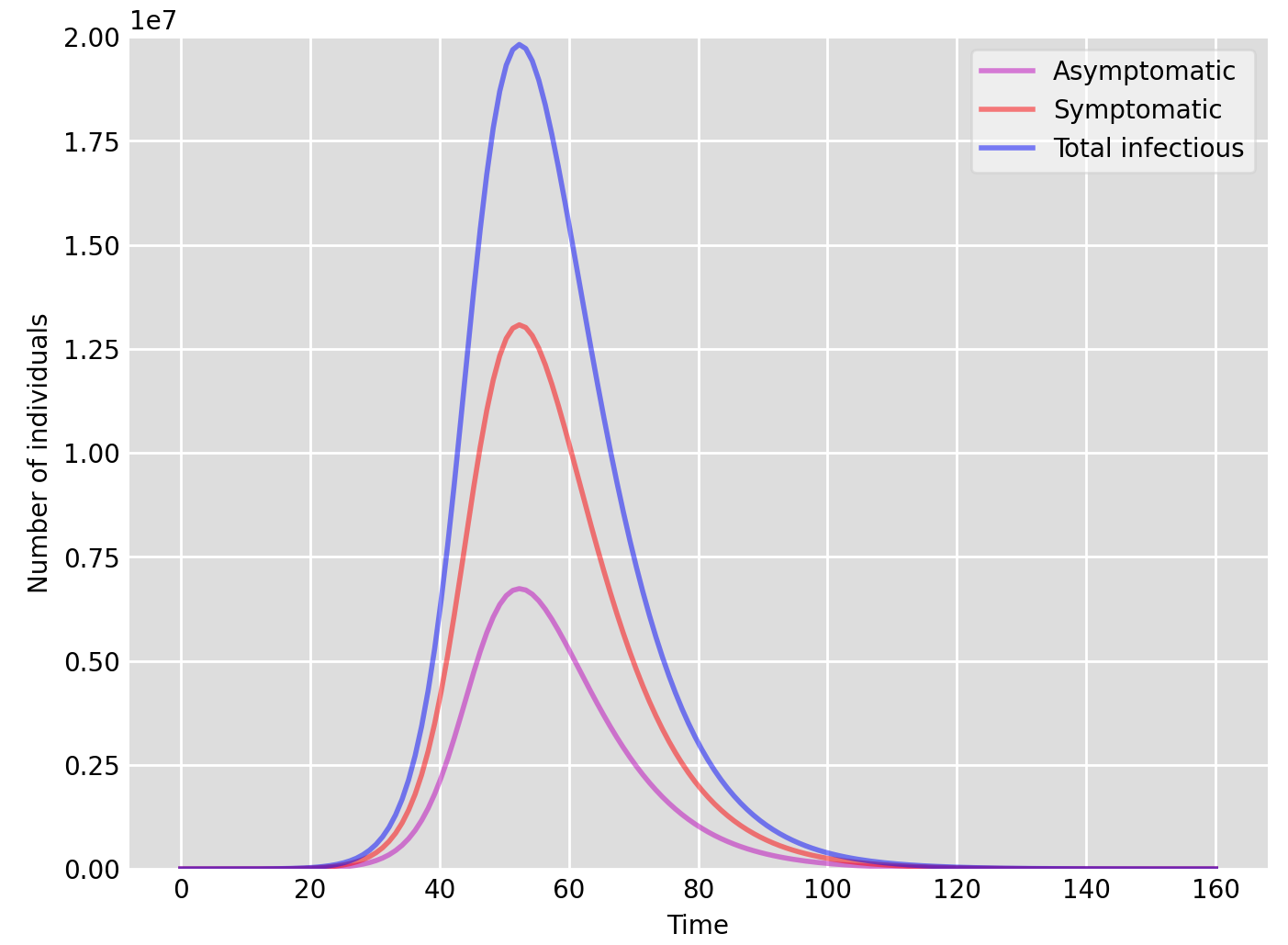}
\caption{Setup 2 - model simulation results for the total number of individuals $A(t),I(t)$  and $A(t)+I(t)$. \label{fig-setup2-total}}
\end{figure}

\begin{figure}[H]
\includegraphics[width=13.5 cm]{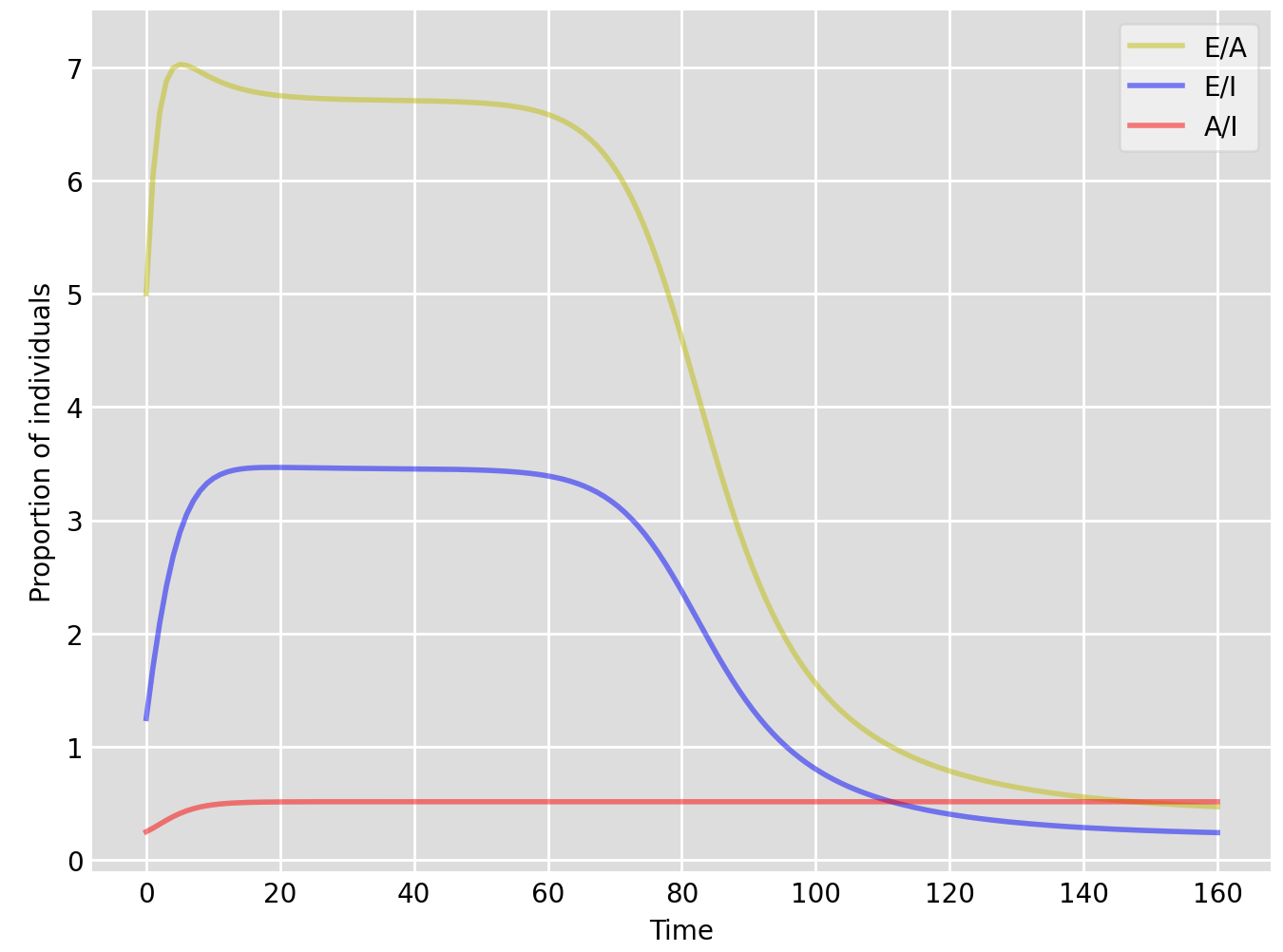}
\caption{Setup 1---model simulation results for the proportions of individuals $\frac{E(t)}{A(t)},\,\frac{E(t)}{I(t)}$ and $\frac{A(t)}{I(t)}$. \label{Setup1-proportions}}
\end{figure}

\begin{figure}[H]
\includegraphics[width=13.5 cm]{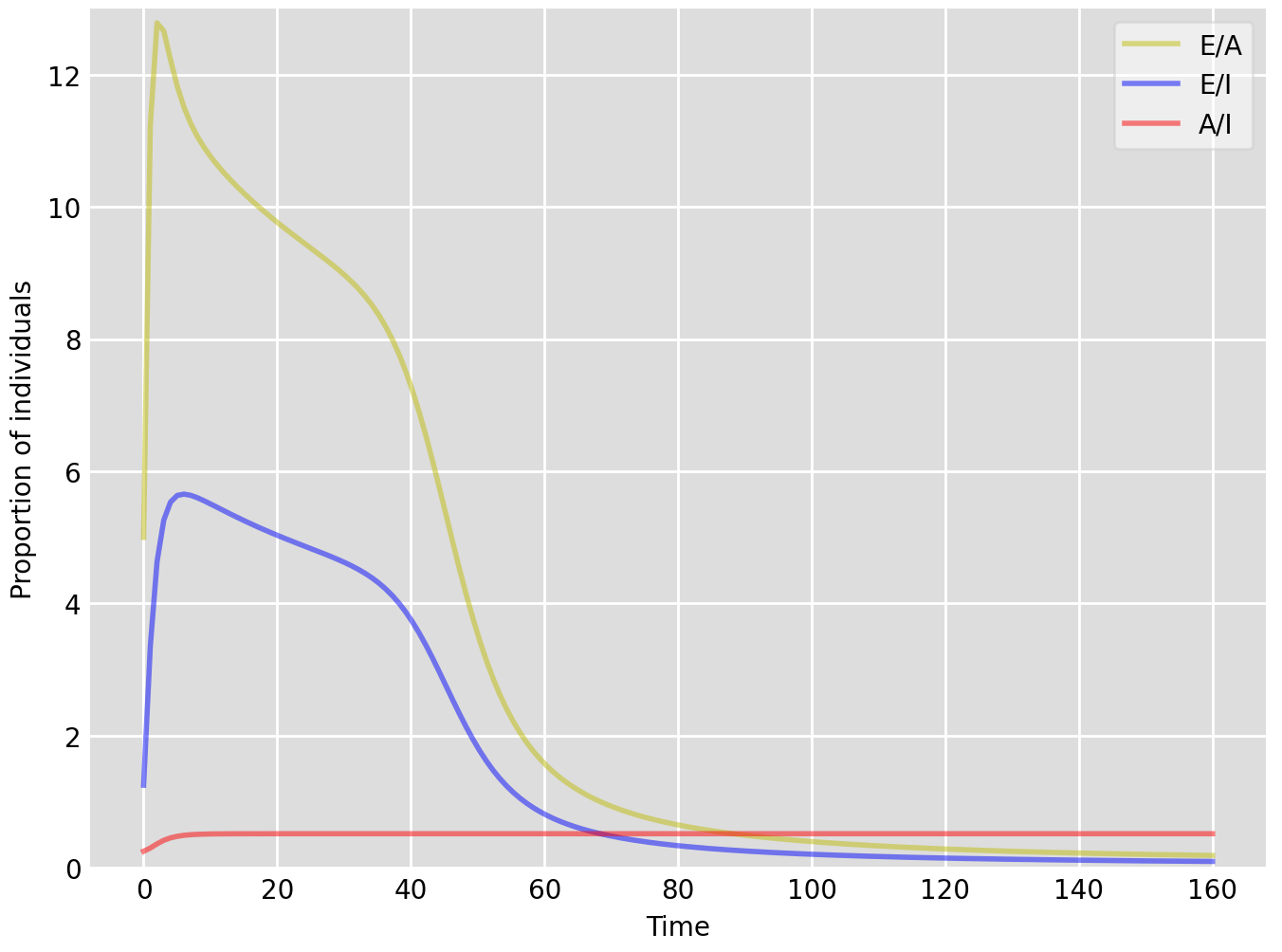}
\caption{Setup 2---model simulation results for the proportions of individuals $\frac{E(t)}{A(t)},\,\frac{E(t)}{I(t)}$ and $\frac{A(t)}{I(t)}$. \label{Setup2-proportions}}
\end{figure}

\begin{figure}[H]
\includegraphics[width=13.5 cm]{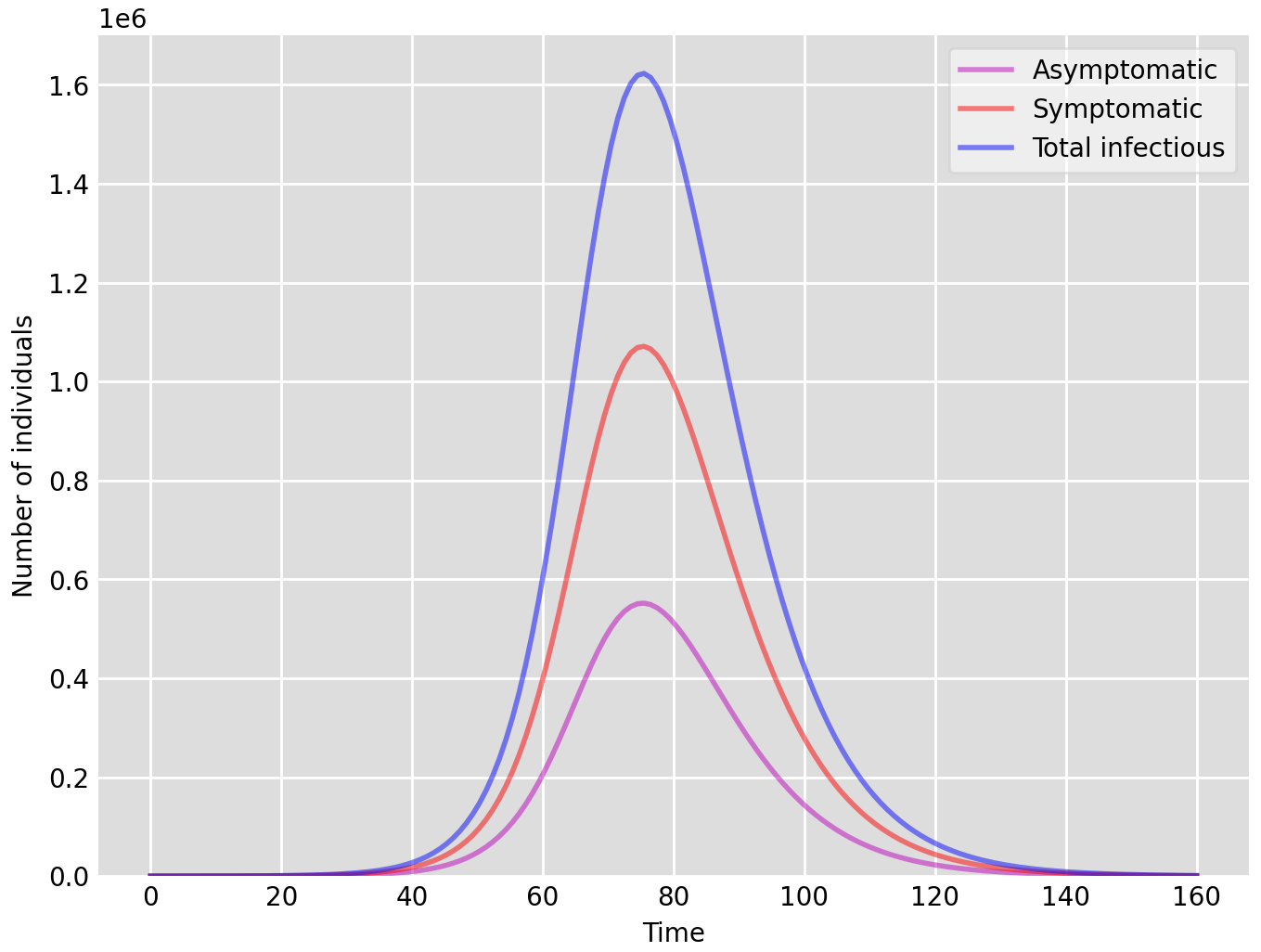}
\caption{Setup 1 reduced $S(0)$ - model simulation results for the total number of individuals $A(t),I(t)$  and $A(t)+I(t)$. \label{fig-setup1reduced-total}}
\end{figure}

\begin{figure}[H]
\includegraphics[width=13.5 cm]{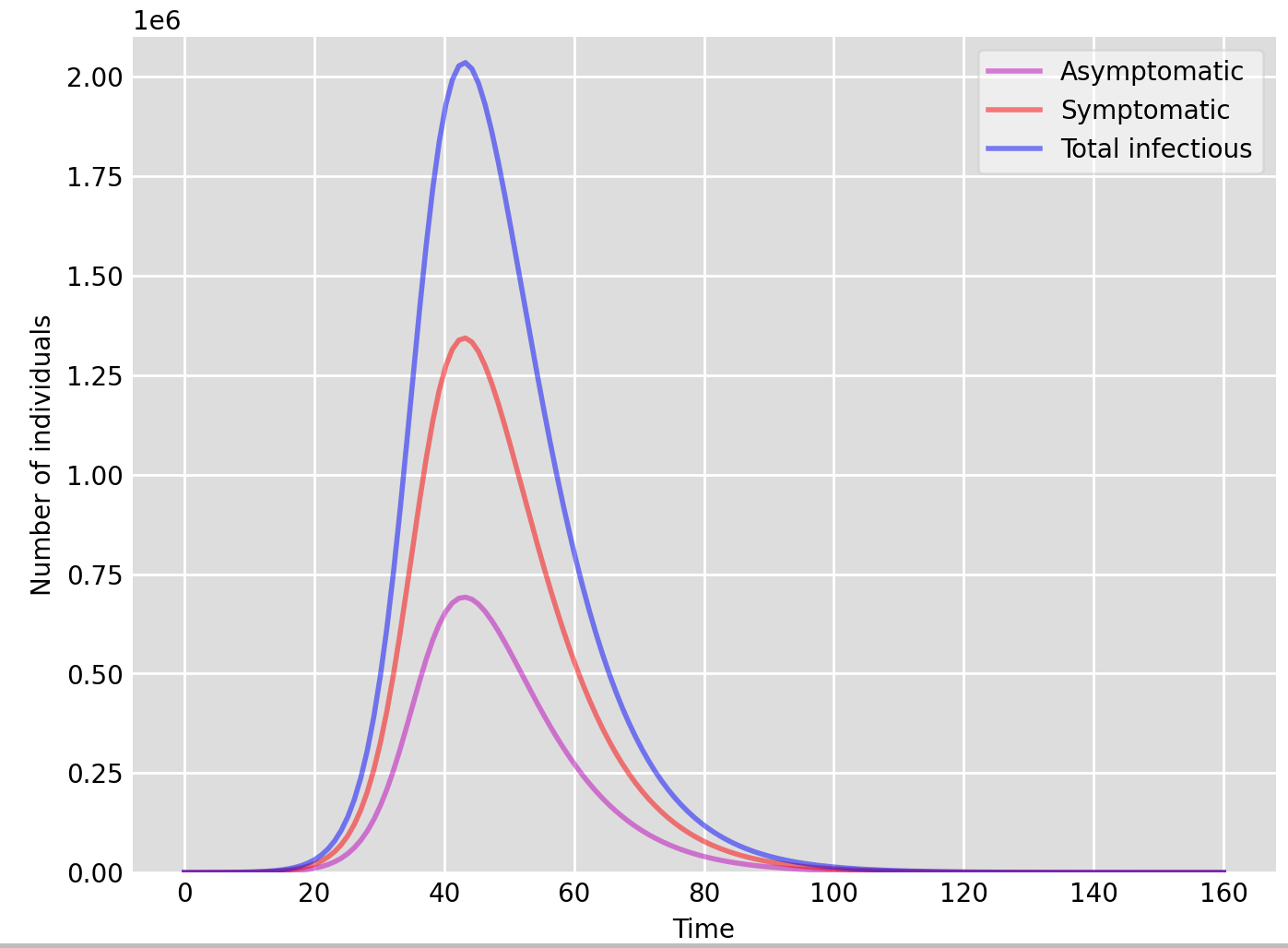}
\caption{Setup 2 reduced $S(0)$ - model simulation results for the total number of individuals $A(t),I(t)$  and $A(t)+I(t)$. \label{fig-setup2reduced-total}}
\end{figure}   

\section{Discussion on $\mathcal{R}_0$---The Basic Reproductive Number}

From a mathematical point of view, a very important feature of the model we proposed here is the lack of continuous differentiability when $I=0$, for $0<\kappa<1$. However, it is clear from the modelling point of view that this is not a practical problem, as we are not interested in solutions starting with $I=0$, when uniqueness may be lost. On the other hand, the lack of continuous differentiability for $0<\kappa<1$ also means that the model cannot be linearised around the infection-free steady state and, hence, the basic reproductive number $\mathcal{R}_0$ cannot be computed following the established next generation method, see \cite{net-rep} for example, for these values of $\kappa$. The problem is that these approaches, see also \cite{Watmough}, try to define and relate the basic reproductive number to the local asymptotic stability of the disease-free steady state, which is somewhat questionable in the first place. It is clear that, as everyone can see in the COVID-19 pandemic, it is meaningless to talk about convergence to a steady state, and the reproductive number changes at every single time instance. Indeed due to the mathematical definition of $\mathcal{R}_0$, e.g., in \cite{net-rep,Watmough}, $\mathcal{R}_0$ only gives you some information at time equals zero, when the disease was introduced into a (completely) susceptible (and well-mixed) population. We refer the interested reader to \cite{Farkas2018}, where a new approach was proposed to define basic reproductive functions for nonlinear models without linearisation. The approach in \cite{Farkas2018} (see also \cite{Farkas2014}) utilises a transformation of the nonlinear problem into a family of linear ones, and makes use of a spectral theoretic result found for example in \cite{Thieme}.

However, for the sake of interest, using the approach from \cite{Watmough}, let us deduce a formula for the basic reproductive number (or rather function, see later) for the case of $\kappa=1$, i.e., when no self-quarantine of symptomatic infectious individuals takes place, and thus our model resembles a classic model. Following \cite{Watmough}, we rearrange the components of the solution vector as $(E(t),A(t),I(t),S(t),R(t))^t$, giving $m=3$ (i.e., the first three components contain infected individuals); and then we have (using the exact same notation as in \cite{Watmough}) 
\begin{equation}\label{operators}
\mathcal{F}=\begin{pmatrix}
\beta_A\,SA+\beta_I\,SI \\
0\\
0\\
0\\
0
\end{pmatrix}, \quad
\mathcal{V}=
\begin{pmatrix}
\alpha\,E \\
-(1-p)\alpha\, E+\gamma_A\, A \\
-\alpha p\, E+\gamma_I\, I \\
\beta_A\, SA+\beta_I\, SI \\
-\gamma_I\, I-\gamma_A\, A
\end{pmatrix}.
\end{equation}

$\mathcal{F}$ above describes the infection process (the recruitment of exposed individuals), while $\mathcal{V}$ describes the transitioning between compartments.
We then compute the derivatives at $(0,0,0,S_*,R_*)^t$ (disease free steady states) as follows ($m=3$):
\begin{equation}\label{operators2}
F=\begin{pmatrix}
0 & \beta_A\, S_* & \beta_I\, S_* \\
0 & 0 & 0 \\
0 & 0 & 0
\end{pmatrix},\quad
V=\begin{pmatrix}
\alpha & 0 & 0\\
-(1-p)\alpha & \gamma_A & 0 \\
-\alpha p & 0 & \gamma_I
\end{pmatrix}.
\end{equation}

The inverse of $V$ is computed as
\begin{equation}\label{operators3}
V^{-1}=
\begin{pmatrix}
\frac{1}{\alpha} & 0 & 0\\
\frac{1-p}{\gamma_A} & \frac{1}{\gamma_A} & 0 \\
\frac{p}{\gamma_I} & 0 & \frac{1}{\gamma_I}
\end{pmatrix},
\end{equation}
which then yields
\begin{equation}\label{operators4}
F\, V^{-1}=
\begin{pmatrix}
\frac{\beta_A\,S_*(1-p)}{\gamma_A}+\frac{\beta_I\, S_* p}{\gamma_I} & \frac{\beta_A\, S_*}{\gamma_A} & \frac{\beta_I\, S_*}{\gamma_I} \\
0 & 0 & 0 \\
0 & 0 & 0 
\end{pmatrix}.
\end{equation}

Hence, the spectral radius of $F\, V^{-1}$ is simply

\begin{equation}\label{netrep}
\mathcal{R}_0=S_*\left(\frac{\beta_A(1-p)}{\gamma_A}+\frac{\beta_Ip}{\gamma_I}\right).
\end{equation}

As we can see from the formula above, $\mathcal{R}_0$ actually can be understood as a function $\mathcal{R}_0(S_*)$ of the susceptible population size $S_*$. Recall that our model \eqref{model} has a continuum family of disease-free steady states of the form $(S_*,0,0,0,R_*)$, where $S_*+R_*=N$ (the total population size). Hence, naturally, for each of these steady states, we get a different value $\mathcal{R}_0(S_*)$.  For example, using the parameter values used in Setup 1, the values for $\mathcal{R}_0(S_*)$ range from $21.28$ for $S_*=N$ to, e.g., $2.128$, corresponding to Setup 1---reduced $S$ ($90\%$ reduction in the susceptible population size due to national lockdown).  There is a bifurcation point on the line of steady states at $S_*=\frac{1}{21.28}N$ (when $\mathcal{R}_0(S_*)=1$), meaning that disease-free steady states $(S_*,0,0,0,R_*)$ are locally asymptotically stable when \mbox{$S_*<\frac{1}{21.28}N\iff \mathcal{R}_0(S_*)<1$,} and those steady states for which $S_*>\frac{1}{21.28}N\iff \mathcal{R}_0(S_*)>1$ are unstable. This qualitative property of model \eqref{model} (again, only for $\kappa=1$) not surprisingly, is somewhat similar to that of the classic Kermack--McKendrick model. The practical interpretation of this result is that according to our model, and using the parameter values as in Setup 1, the pandemic could have been avoided theoretically (without quarantining, i.e., $\kappa=1$) by isolating $\frac{20.28}{21.28}\%\approx 95.3\% $ of the susceptible population (or in other words, an effective susceptible population size of $\approx 4.7\%\,N$), a percentage clearly hard to achieve.

\section{Outlook}

Our goal here was to introduce a basic compartmental mathematical model, incorporating a new infection term, to model (short term, horizontal) disease transmission dynamics, by focusing on two key aspect: the role of asymptomatic infectious individuals and (self)-quarantine. From a mathematical point of view, our model is interesting, as, for example, the basic reproduction number cannot be computed for certain parameter values following established mathematical approaches. We parametrised our model using studies from the COVID-19 literature. However, we do not want to overstate our conclusions for COVID-19 as we only performed a limited number of simulations, and there are significant challenges and unknowns around data collection (for example test data) and analysis; important issues on which we did not focus here.

There are various natural modifications and extensions of our model, some of which we mention here briefly. For example, it is clear that all of the model parameters can be made explicitly time-dependent. It could be interesting to replace $\kappa$ with a time-dependent function $\kappa(t)$, to allow us to model the effects of COVID-19 fatigue, that is weakening adherence to self-quarantine of symptomatic infectious individuals. This would be particularly relevant when focusing on subsequent waves. 
If we want to look beyond the first wave, then naturally we may introduce and study the impact of waning immunity. From the modelling perspective, this would mean individuals re-entering the $S$ compartment from the $R$ compartment. The introduction of time-dependent model parameters would, in principle, allow to model long-term disease dynamics and allow for periodic waves of outbreaks. Naturally acquired COVID-19 immunity (e.g., via infection) can be incorporated by simply introducing a constant (less than 1) multiplying the parameter values $p$ and $1-p$. This would be important when modelling subsequent waves as a number of people have acquired immunity through the first wave. It could be interesting to compare simulation outputs for infectious individual numbers to that of test data during the 2020/2021 winter wave, when mass testing was available and in the UK we have seen 10--15 times the number of daily positive tests compared to the first wave, despite that, a significant proportion of the population may have acquired immunity during the first wave.

\end{document}